\documentclass[aps,pra,twocolumn,amsfonts,amssymb,amsmath,showpacs,
floatfix,nofootinbib,citesort]{revtex4-2}
\usepackage{mathrsfs}
\usepackage{amsfonts}
\usepackage{amstext}
\usepackage{amsmath}
\usepackage{amssymb}
\usepackage{bm}
\usepackage{CJK}
\usepackage{bbm}
\usepackage[dvips]{graphicx}
\def\qed{\leavevmode\unskip\penalty9999 \hbox{}\nobreak\hfill
	\quad\hbox{\leavevmode  \hbox to.77778em{%
			\hfil\vrule   \vbox to.675em%
			{\hrule width.6em\vfil\hrule}\vrule\hfil}}
	\par\vskip3pt}

\usepackage{amssymb}
\usepackage{graphicx}
\usepackage{graphics}
\usepackage{amsmath}
\usepackage{amsthm}
\usepackage{color}
\usepackage{booktabs}
\usepackage{makecell}
\usepackage{dsfont}
\usepackage{textcomp}
\usepackage{threeparttable}
\definecolor{darkred}  {rgb}{0.5,0,0}
\definecolor{darkblue} {rgb}{0,0,0.5}
\definecolor{darkgreen}{rgb}{0,0.5,0}
\usepackage{hyperref}
\hypersetup{
	pdftitle = {QRT Proposal},
	pdfauthor = {},
	colorlinks = true,
	urlcolor  = blue,         
	linkcolor = red,     
	citecolor = blue,    
	filecolor = darkred       
}
\usepackage{mathtools}
\def\ra{\rangle}
\def\la{\langle}

\newtheorem{theorem}{Theorem}

\newtheorem{cor}[theorem]{Corollary}

\newcommand{\bea}{\begin{eqnarray}}
	\newcommand{\eea}{\end{eqnarray}}
\newcommand{\be}{\begin{equation}}
	\newcommand{\ee}{\end{equation}}
\newcommand{\ba}{\begin{equation}\begin{aligned}}
		\newcommand{\ea}{\end{aligned}\end{equation}}

\newcommand{\beax}{\begin{eqnarray*}}
	\newcommand{\eeax}{\end{eqnarray*}}
\newcommand{\bex}{\begin{equation*}}
	\newcommand{\eex}{\end{equation*}}

\theoremstyle{remark}

\newtheorem{example}{Example}

\def\be{\begin{equation}}
	\def\ee{\end{equation}}

\newcommand{\mE}{\mathcal{E}}

\newcommand{\mH}{\mathcal{H}}

\newcommand{\mT}{\mathcal{T}}

\newcommand{\mP}{\mathcal{P}}

\newcommand{\mS}{\mathcal{S}}

\newcommand{\lr}{\rangle\langle}

\newcommand{\tr}{{\rm Tr}}

\newcommand{\bra}[1]{\langle #1|}
\newcommand{\ket}[1]{|#1\rangle}

\newcommand{\Fi}{\it{\rm Fi}}
\newcommand{\rmi}{{\rm i}}

\begin{document}
	

\preprint{APS/123-QED}
\begin{CJK*}{GB}{gbsn}
\title{Quantum-state texture measure and texture transformation\\}


\author{Yufan Lin}
\author{Yu Guo}
\email{guoyu3@aliyun.com}		
\author{Fei He}

\affiliation{School of Mathematical Sciences, Inner Mongolia University, Hohhot, Inner Mongolia 010021, People's Republic of China}
\affiliation{Inner Mongolia Key Laboratory of Mathematical Modeling and Scientific Computing, Inner Mongolia University, Hohhot, Inner Mongolia 010021, People's Republic of China}

\author{Shuanping Du}

\affiliation{School of Mathematical Sciences, Xiamen University, Xiamen, Fujian, 361000, People's Republic of China}


\begin{abstract}
Quantum-state texture quantifies structural irregularities of a state in a selected basis and has emerged as a valuable resource for gate characterization in universal circuits. Consequently, a class of contractive distance based texture measures and a framework for constructing convex-roof extended texture measure have been proposed. Here, we extend this theory in several directions. We examine the convertibility of a pure state to any state under texture-free operations, and fully solve the deterministic transformation problem between any two states in the qubit case. We then propose two classes of texture measures: convex-function-based measures and texture cost defined via minimal cost of texture contained in pure states. Moreover, we show that every texture monotone gives rise to a non-negative function which admits the same properties as the function used in the convex-roof extension---thus providing a converse to that construction. Our work thereby offers a relatively comprehensive resource-theoretic formulation of quantum-state texture.

\end{abstract}

\maketitle
\end{CJK*}		


\section{Introduction}


Quantum resource theory (QRT) has established a unified operational framework for the characterization of nonclassicality in quantum systems. Under this umbrella, features including entanglement~\cite{Horodecki2009rmp,Vidal2000jmo,Ekert1991prl,Guo2023njp}, coherence~\cite{Baumgratz2014prl,Streltsov2017prl,Du2015qic}, asymmetry~\cite{Gour2009pra}, and imaginarity~\cite{Hickey2018jpa,Du2025pra,Zheng2025jpa,Wu2021pra,Wu2025ctp} are reinterpreted not as incidental mathematical attributes, but as physical resources amenable to consumption, manipulation, and quantification---thereby offering operational advantages in tasks such as quantum key distribution~\cite{Ma2019pra,Curty2004prl}, quantum computation~\cite{Hillery2016pra,Shi2017pra}, and quantum thermodynamics~\cite{Masanes2017nc,Brandao2015}. A generic QRT is delineated by three essential components: the specification of free states, the admissible class of free operations, and an appropriate resource measure. Among these, the construction of reliable resource quantifiers is of paramount importance, as it underpins the assessment of resourcefulness, the study of state convertibility, and the elucidation of operational significance.

In this context, Parisio recently introduced  the notion of quantum-state texture (QST) as a novel type of quantum resource~\cite{Parisio2024prl}. Given a fixed reference basis $\{|i\ra\}_{i=0}^{d-1}$ in the state space $\mH$, a quantum state can be represented by its density matrix, which may be visualized as a three-dimensional landscape, where the matrix elements determine the ``height'' of the surface. In general, this landscape exhibits irregularities, referred to as ``texture'', while the only textureless state corresponds to the uniform superposition state  $\tau_f=\frac{1}{d}\sum_{i,j=0}^{d-1}|i\lr j|$. Very recently, two distinct generalizations of quantum-state texture have been developed: Chu \textit{et al.}~\cite{Chu2026aqt} introduced quantum-state block texture and established its resource-theoretic framework, while Greenwood \textit{et al.}~\cite{Greenwood2026pra} generalized quantum-state texture by allowing any pure state to be chosen as the minimum-resource state. Vieira \textit{et al.}~\cite{Vieira2026arvix} developed a general dynamical theory of quantum-state texture under arbitrary quantum channels and experimentally demonstrated its dynamics and operational applications. In addition, Huang \textit{et al}.~\cite{Huang2025qip} explored quantum-state texture dynamics in the system of uniformly accelerated atoms interacting with a massive scalar field.

The significance of QST lies in both its conceptual novelty and its potential applications. Parisio established a resource-theoretic framework for QST and introduced state rugosity as a computable quantifier. In particular, QST was shown to enable the identification of unknown quantum circuit layers containing CNOT gates without full state tomography or ancillary systems. C\'eleri \textit{et al}.~\cite{Celeri2026arvix} found that quantum-state texture signals dynamical quantum phase transitions.

A central problem in quantum resource theories concerns the quantum state transformations, i.e., the manipulation of physical resource, which asks whether a quantum state $\rho$ can be transformed into another state $\sigma$ under a given set of free operations~\cite{Chitambar2019rmp,Regula2022prl}. The answer to this question
determines ``What tasks may be accomplished using a given physical resource''. For the case of entanglement and coherence, this problem has been extensively studied in various settings, including deterministic transformations, probabilistic transformations, and asymptotic transformations \cite{Du2015qic,Shi2017sr,Wu2020npj,Yu2020pra,Du2024qic,Nielsen1999prl,Du2015pra,Nielsen2001qic,Ferrari2023cmp,Datta2023prl}. In Ref.~\cite{Wu2021pra}, the conversion of imaginarity was discussed. They firstly obtained, for any quantum resource, an upper bound of the maximal probability for converting  two arbitrary states in terms of any associated strong resource monotone. By this upper bound, a maximal probability for probabilistic  transformations between pure states under real operations was proposed. Furthermore, they derived a necessary and sufficient condition for deterministic transformations between two arbitrary qubit states, yielding a complete description of state conversion in the qubit regime. This motivates us to discuss such an issue for the case of texture.

A valid quantum-state texture measure, denoted by $\mT$, should satisfy three fundamental requirements: non-negativity, monotonicity, and convexity.  
Based on the above theoretical framework, a variety of quantum-state texture measures have been proposed~\cite{Wang2025pra,Zhang2025pla,Cao2026jpa,Muthuganesan2026pla,Patra2026pra,Chen2026jpa,Cui2025arvix}. Representative examples include the trace-distance texture measure, fidelity-based texture measure, Bures distance texture measure, and geometric texture measure introduced by Wang \emph{et al.}~\cite{Wang2025pra}, the weight-based texture measure by Zhang \emph{et al.}~\cite{Zhang2025pla}, Tsallis relative entropy texture measure ~\cite{Zhang2025pla,Cao2026jpa}, the Hellinger distance texture measure, Jensen-Shannon divergence texture measure, and Wigner-Yanase-Dyson (WYD) skew information based texture measure developed by Muthuganesan~\cite{Muthuganesan2026pla}, and the convex-roof extended texture measure by Cui~\cite{Cui2025arvix}. All these measures have counterparts in the context of other convex quantum resource theories. Note that in the context of texture, there is only one free state, so we expect to find new approaches to constructing texture quantifiers.

Another type of resource measure is defined as the minimal cost of pure resource state from which the given state can be obtained under free operations. For example, such measures include the entanglement measure in terms of LOCC on pure states~\cite{Yu2021cpb,Shi2021anpb}, coherence cost~\cite{Yu2020pra,Winter2016prl,Zhao2018prl} and imaginarity cost~\cite{Du2025pra}. We discuss in this paper the counterpart for texture, which we call  texture cost.

The rest of this paper is organized as follows. In Sec.~\ref{preli}, we briefly introduce the basic concepts of the texture resource theory and review the existing measures of quantum-state texture. In Sec.~\ref{state}, we investigate the state transformation problems under the texture-free operations. In particular, we discuss the qubit case in detail. Sec.~\ref{build} presents two kinds of new approaches for defining texture measures and analyzes the relation between a general texture monotone and the convex-roof extended texture monotone. Finally, we conclude this work in Sec.~\ref{concl}.


\section{Preliminary: Texture and texture measure}\label{preli}


\subsection{Basic concepts}

Throughout this paper, we denote by $\mS:=\mS(\mH)$ the set of all quantum states on the Hilbert space $\mH$ with $\dim\mH=d\geqslant2$. 
For a given orthonormal basis $\{|i\ra\}_{i=0}^{d-1}$ in $\mH$,
\beax
\tau_f=|f_1\lr f_1|=\frac{1}{d}\sum_{i,j=0}^{d-1}|i\lr j|	
\eeax
is the only one textureless state~\cite{Parisio2024prl} in this basis, where $|f_1\ra=\frac{1}{\sqrt{d}} \sum_{i=0}^{d-1}|i\ra$. Any state in $\mS-\{\tau_f\}$ is a resource state or called a textured state in the resource theory of texture.

The free operation of texture~\cite{Parisio2024prl}, which we call a texture-free operation hereafter, is the completely positive and trace-preserving (CPTP) map $\mE$ that also satisfies $K_j|f_1\ra\propto |f_1\ra$ for all $j$, where $K_j$'s are the Kraus operators corresponding to $\mE$, i.e., $\mE(\rho)=\sum_jK_j \rho K_j^\dagger$ with $\sum_jK_j^\dagger K_j = I$, $I$ denotes the identity operator on $\mH$. Clearly, $\tau_f$ is a fixed point of $\mE$, namely, $\mE(\tau_f) = \tau_f$.

A nonnegative function $\mT: \mS\rightarrow[0, +\infty)$ is called a texture measure if $\mT$ satisfies the following
conditions (T1)-(T3)~\cite{Parisio2024prl}: 

(T1) $\mT(\tau_f)=0$;

(T2) Monotonicity: $\mT(\mE(\rho))\leqslant\mT(\rho)$ for any $\rho\in\mS$ and any texture-free operation $\mE$;

(T3) Convexity: $\mT(\sum_jp_j\rho_j)\leqslant\sum_jp_j{\mT}(\rho_j)$ for any ensemble $\{p_j,\rho_j\}$.

Following the established theories of quantum coherence measures and quantum imaginarity measures, two additional desirable properties were put forward in Ref.~\cite{Zhang2025pla}:

(T4) Strong monotonicity: for any texture-free operation $\mE$ with Kraus operators  $K_j$'s, $\mT(\rho)\geqslant\sum_jp_j\mT(\sigma_j)$, where $\sigma_j={K_j\rho K_j^\dagger}/{p_j}$, $p_j=\tr(K_j\rho K_j^\dagger)$.

(T5) Direct-sum additivity: $\mT(p\rho_1\oplus(1-p)\rho_2)= p\mT(\rho_1)+(1-p)\mT(\rho_2)$, 
where $\rho_1\in\mS(\mH_1)$ and $\rho_2\in\mS(\mH_2)$. The reference basis of $\mH_1\oplus \mH_2$ is taken as the direct-sum $\{|i\ra\oplus\mathbf{0}_2,\mathbf{0}_1 \oplus |j\ra\}$ of the fixed bases $\{|i\ra\}$ of $\mH_1$ and $\{|j\ra\}$ of $\mH_2$, where $\mathbf{0}_{1,2}$ denotes the zero vector in $\mH_{1,2}$.

Items (T1)-(T5) are indeed the counterparts of the corresponding properties for the coherence and imaginarity measures, denoted by (C1)-(C5) and (I1)-(I5), respectively~\cite{Xue2021qip,Yu2016pra1,Xu2020pra}. These properties are not independent but are connected through specific logical implications. In particular, any measure that satisfies (C1)-(C4) (resp. (I1)-(I4)) also fulfills (C5) (resp. (I5)). Conversely, a measure satisfying (C1), (C2), and (C5) (resp. (I1), (I2), and (I5)) necessarily satisfies both (C3) and (C4) (resp. (I3) and (I4)). However, it is not valid for texture measure~\cite{Zhang2025pla}: (i) even if a QST measure satisfies (T1)-(T4), it does not guarantee that it also satisfies (T5); (ii) if a measure $T$ satisfies (T1), (T2), and (T5), then strong monotonicity (T4) may hold even when (T3) does not. Hereafter, if a texture measure satisfies item (T4), we call it a texture monotone.

\subsection{Texture measures}

In quantum resource theory, one often uses a class of contractive distances to quantify the ``amount'' of resource contained in the states~\cite{Baumgratz2014prl,Vedral1997prl,Streltsov2017rmp}. Let $D(\rho,\sigma)$ be a distance or divergence defined on quantum states. If for every quantum channel, i.e., every CPTP map $\mE$, it satisfies
\beax
D(\rho,\sigma)\geqslant D\bigl(\mE(\rho),\mE(\sigma)\bigr),
\eeax
$D$ is said to be contractive under quantum operations~\cite{Vedral1997prl}.
In the following, we introduce texture measures constructed by this kind of contractive distances (or divergences), such as the trace distance measure~\cite{Wang2025pra}, fidelity-based measure~\cite{Wang2025pra}, Bures measure~\cite{Wang2025pra}, Tsallis relative entropy measure~\cite{Zhang2025pla,Cao2026jpa}, Jensen-Shannon divergence measure~\cite{Muthuganesan2026pla,Cao2026jpa}, and Hellinger distance measure~\cite{Muthuganesan2026pla}.

The trace distance measure is defined by~\cite{Wang2025pra}
\be\label{trace-distance}
\mT_{tr}(\rho)=\dfrac{1}{2}\left\|\rho -\tau_f\right\|_{\tr},
\ee
where $\|A\|_\tr=\tr\sqrt{A^\dagger A}$ is the trace norm of the matrix $A$.
Fidelity-based measure is~\cite{Wang2025pra}
\be
\mT_{\Fi}(\rho)=1-F(\rho,\tau_f),
\ee
where $F(\rho,\sigma)=\left[\tr\left( \sqrt{\rho^{1/2} \sigma \rho^{1/2}} \right) \right]^2$ is the Uhlmann fidelity.
The Bures measure is defined via the fidelity~\cite{Wang2025pra}, i.e.,
\be
\mT_B(\rho)=2\left(1-\sqrt{F(\rho, \tau_f)} \right).
\ee
Relative entropy measure~\cite{Wang2025pra}
\be\label{relative-entropy}
\mT_r(\rho)=S(\rho\|\tau_f).
\ee
Tsallis relative entropy measure~\cite{Zhang2025pla,Cao2026jpa}
\be
\mT_{\mu}^{TS}(\rho)=\dfrac{1-\la f_1|\rho^{\mu}|f_1\ra}{1-\mu}, \quad \mu \in (0,1),
\ee
and the Jensen-Shannon divergence measure~\cite{Muthuganesan2026pla,Cao2026jpa}
\be
\mT_J(\rho)={J}(\rho,\tau_f)=S\left(\dfrac{\rho+\tau_f}{2} \right)-\dfrac{1}{2}\left[S(\rho)+S(\tau_f) \right]
\ee
are entropy-based measures, where $S(\cdot)$ denotes the von Neumann entropy, 
$S(\rho\|\sigma)=\tr(\rho\log_2\rho-\rho\log_2\sigma)$ is the relative entropy.
$\mT_J(\rho)$ satisfies~\cite{Muthuganesan2026pla}
\be\label{bound}
\frac{[\mT_{tr}(\rho)]^2}{2\ln 2}\leqslant \mT_J(\rho)\leqslant H_2\left( \frac{1+\sqrt{F(\rho,\tau_f)}}{2} \right),
\ee
where
$H_2(x)=-x\log_2x-(1-x)\log_2(1-x)$
is the binary entropy function.
Another contractive distance is the Hellinger distance, which gives the Hellinger distance measure~\cite{Muthuganesan2026pla}
\be
\mT_{he}(\rho)=\tr\left( \sqrt{\rho}-\sqrt{\tau_f}\right)^2.
\ee
$\mT_{he}(\rho)$ is bounded by~\cite{Muthuganesan2026pla}
\be 
2\mT_B(\rho)\leqslant\mT_{he}(\rho)\leqslant2\mT_{tr}(\rho).
\ee

There are other types of measures that are not deduced from the contractive distance. For example, the state rugosity measure $\mT_R$~\cite{Parisio2024prl}, the geometric texture measure $\mT_g$~\cite{Wang2025pra}, and the weight of quantum state texture $\mT_w$~\cite{Zhang2025pla}. Recall that,
\bea\label{rug}
\mT_R(\rho)&=&-\ln\la f_1 |\rho|f_1\ra,\\
\mT_g(|\psi\ra)&=&1-|\la f_1 |\psi\ra|^2,
\eea
and for a mixed state
\be\label{ge1}
\mT_g(\rho)=\min_{\{p_i,|\psi_i\ra\}}\sum_ip_i\mT_g(|\psi_i\ra),
\ee
where the minimization is taken over all pure-state decompositions $\{p_i, |\psi_i\ra\}$ of $\rho$.
It was proven that~\cite{Wang2025pra,Zhang2025pla} 
\bea\label{T_g-F}
\mT_g(\rho)=\mT_{\Fi}(\rho)=1-\la f_1|\rho|f_1\ra\geqslant&[\mT_{tr}(\rho)]^2.
\eea
Indeed, $\mT_g$ is a special case of the convex-roof extended measure~\cite{Cui2025arvix}. Let $h:[0,1]\rightarrow[0, +\infty)$ be a function that satisfies (i) $h(1)=0$, 
(ii) $h$ is monotonic decreasing, and (iii) $h$ is concave, i.e., $h(\lambda {x}+(1-\lambda)
{y})\geqslant\lambda h({x})+(1-\lambda)h({y})$
for all $\lambda\in[0,1]$ and all $x,y\in [0,1]$.
For any pure state $|\psi\lr\psi|$, let~\cite{Cui2025arvix}
\be\label{fun1}
{\mT}(|\psi\ra)=h(|\la f_1|\psi\ra|^{2})
\ee
and
\be\label{fun2}
{\mT}_F(\rho)=\min_{p_i,|\psi_i\ra}\{\sum_ip_i{\mT}(|\psi_i\ra):\ \rho=\sum\limits_i
p_i|\psi_i\lr\psi_i|\}
\ee	
for any mixed state \( \rho \). $\mT_F$ is called the convex-roof extension of $\mT$. Taking $h(x)=1-x$, the corresponding convex-roof extended measure is just $\mT_g$.
It was shown in Ref.~\cite{Cui2025arvix} that, for any function $h:[0,1]\rightarrow[0, +\infty)$ satisfying items (i)-(iii), $\mT_F$ constitutes a valid texture monotone [for any pure state $|\psi\ra$, we let $\mT_F(|\psi\ra)=\mT(|\psi\ra)$].

The weight of quantum state texture is defined as~\cite{Zhang2025pla}
\be\label{weight}
\mT_w(\rho)=\min\limits_{\tau} \left\{s\geqslant0\,\Big|\,\rho=(1-s)\tau_f +s\tau, \tau\in\mS\right\},
\ee
where the minimum is taken over all possible $\tau$ such that $\rho=(1-s)\tau_f +s\tau$ with $0\leqslant s\leqslant 1$.
It was proven that~\cite{Zhang2025pla}
\be 
\mT_w(\rho)\geqslant\mT_{\Fi}(\rho).
\ee
Robustness measure~\cite{Wang2025pra} 
\be
\mT_{rob}(\rho)=\min_{\sigma}\left\{s\geqslant 0\,\bigg|\,\frac{\rho+s\sigma}{1+s}=\tau_f \right\},
\ee
where the minimization runs over all quantum states $\sigma$ for which the convex mixture of $\rho$ and $\sigma$ yields a textureless state. For most quantum states, both $\mT_r$ and $\mT_{rob}$ are unbounded~\cite{Wang2025pra}. 

Very recently, Muthuganesan proposed the
WYD skew information based texture quantifier~\cite{Muthuganesan2026pla}
\bea 
\mT_{\alpha}^{\mathrm{skew}}(\rho) &= &\tr(\rho\tau_f^2)-\tr(\rho^\alpha\tau_f\rho^{1-\alpha}\tau_f)\nonumber\\
&=&\la f_1|\rho|f_1\ra-\la f_1|\rho^\alpha|f_1 \ra \la f_1|\rho^{1-\alpha}|f_1\ra,
\eea
where $0<\alpha<1$.
But $\mT_{\alpha}^{\mathrm{skew}}$ fails to satisfy monotonicity condition~\cite{Muthuganesan2026pla}. It only satisfies (T2) with CPTP $\mE$ satisfying $\mE^\dagger(\tau_f)=\tau_f$.


\section{State transformations via texture-free operations}\label{state}


We now review the upper bound of the success probability of the state transformations under free operations, within the general framework of quantum resource proposed in~\cite{Wu2021pra}. 
A valid resource measure $R$ is usually required to be nonincreasing under free operations. 
When stochastic transformations are considered, one often imposes the strong monotonicity. Suppose that a free operation acts on an initial state $\rho$ and outputs an ensemble 
$\{p_j,\sigma_j\}$, where
\beax
\sigma_j=\frac{K_j\rho K_j^\dagger}{p_j},\quad 
p_j=\tr(K_j\rho K_j^\dagger),
\eeax
and $\{K_j\}$ are free Kraus operators of the free operation. $R$ is called a strong resource monotone~\cite{explain} if
\bea\label{monotone}
R(\rho)\geqslant \sum_jp_j R(\sigma_j)
\eea
as item (T4) and $R$ is convex additionally, i.e.,
\bea\label{convex}
R\left(\sum_j p_j\rho_j\right)\leqslant \sum_j p_jR(\rho_j)
\eea
for any ensemble $\{p_j, \rho_j\}$. Clearly, if a non-negative function $R$ satisfies items~\eqref{monotone} and~\eqref{convex}, it is non-increasing under any free operation. {For a unified description, throughout this paper we refer to a measure satisfying non-negativity, monotonicity, strong monotonicity and convexity as a resource monotone. These two properties also lead to a general upper bound on the success probability of resource-state conversion under the trace-nonincreasing free operations~\cite{Wu2021pra}:
\be\label{tra1}
P(\rho\to\sigma)\leqslant \min\left\{\frac{R(\rho)}{R(\sigma)},\,1\right\}
\ee
for any associated resource monotone $R$.
Here, the success probability (also called the maximal probability in Ref.~\cite{Wu2021pra}) for converting $\rho$ into $\sigma$ is defined by
\[
P(\rho \to \sigma) = \max_{\{K_j\}} \left\{ \sum_j p_j : \sigma = \frac{\sum_j K_j \rho K_j^\dagger}{\sum_j p_j} \right\}
\]
with probabilities $p_j = \operatorname{Tr}(K_j \rho K_j^\dagger)$, where the maximum is taken over all trace-nonincreasing free operations, and the corresponding Kraus operators are denoted by $\{K_j\}$. Here, with some abuse of terminology, we say a completely positive map $\mE$ on $\mS$ is a trace-nonincreasing free operation, if the associated Kraus operators admit $\sum_jK_j^\dag K_j\leq I$. The existence of free operation $\mE$ such that $\mE(\rho)=\sigma$ implies that $P(\rho \to \sigma) = 1$.

\subsection{The transformation probability from pure states to mixed states}\label{trans}	

In the context of imaginarity, the optimal probability for a pure-state conversion $|\psi\rangle \to |\phi\rangle$ via real operations is known to be~\cite{Wu2021pra}
\[
P(|\psi\rangle \to |\phi\rangle) = \min \left\{ \frac{1 - |\langle\psi^*|\psi\rangle|}{1 - |\langle\phi^*|\phi\rangle|}, 1 \right\},
\]
where $|\psi^*\ra$ denotes the complex conjugation of $|\psi\ra$ in the given reference basis of imaginarity. We will show below that a parallel result applies to texture, and moreover, this framework admits further refinement.

\begin{theorem}\label{th1}
The optimal probability with which $|\psi\ra$ can be transformed into $\rho$ via texture-free operations is given by
\be\label{tra2}
P(|\psi\ra\to\rho)=\min\left\{\dfrac{\mT_g(|\psi\ra)}{\mT_g(\rho)}, \, 1\right\}.
\ee		
\end{theorem}

\begin{proof}
Let
$\rho=\sum_iq_i|\phi_i\lr\phi_i|$
for some ensemble $\{q_i, |\phi_i\ra\}$ and let
\beax p=\min\left\{ \dfrac{\mT_g(|\psi\ra)}{\mT_g(\rho)},\, 1 \right\}, ~
\alpha=\la f_1|\psi\ra,~ x=1-|\alpha|^2,
\eeax
then 
\beax 
|\psi\ra=\alpha|f_1\ra+\sqrt{x}|u\ra
\eeax
for some $|u\ra$ with $\la f_1|u\ra=0$.
For each $|\phi_i\ra$, we write
$\beta_i=\la f_1|\phi_i\ra$ and $y_i=1-|\beta_i|^2$,
then
\beax
|\phi_i\ra= \beta_i|f_1\ra + \sqrt{y_i}|v_i\ra
\eeax
for some $|v_i\ra$'s with $\la f_1|v_i\ra=0$.
It follows from Eq.~\eqref{T_g-F} that
\beax
t:= T_g(\rho)=1-\la f_1|\rho|f_1\ra=\sum_iq_iy_i.
\eeax

\textit{Case 1: $x = 1$.}
In this case, $|\psi\ra=|u\ra$. We consider the representation of $|\psi\ra$ in the basis of $\{|f_1\ra, |f_2\ra, |f_3 \ra \cdots |f_d\ra\}$, i.e., $|\psi\ra=\sum_{i=1}^d c_i|f_i\ra$. Here, the states $\{|f_k\ra\}_{k=1}^d$  are chosen as the Fourier states introduced in Ref.~\cite{Parisio2024prl}. $|\psi\ra=\sum_{i=2}^d c_i|f_i\ra$ indeed since $\la f_1|\psi\ra=0$.
Note again that $|f_1\ra$ is orthogonal to $|u\ra (=|\psi\ra)$, we let $\{|f_1\ra, |\psi\ra, |e_3\ra \cdots |e_d\ra\}$ be an orthonormal basis of the state space $\mH$.
Let
\be
U=|f_1\lr f_1|+|f_2\lr\psi|+\sum_{i=3}^d|f_i\lr e_i|.
\ee
Then $U|f_1\ra=|f_1\ra$. So $\mE_U(\cdot) = U(\cdot)U^\dagger$ is a free operation, and meanwhile, $U|\psi\ra=|f_2\ra$. Thus, for the target state $\rho$, one can proceed as follows. First, we apply the texture-free operation $\mE_U$ to transform $|\psi\ra$ into $|f_2\ra$.
Then, by the texture-free operation in Ref.~\cite{Parisio2024prl}, $|f_2\ra$ can be transformed into $\rho$. Namely, when $x=1$,
$P(|\psi\ra\to\rho)=1$.

\textit{Case 2: $x<1$ and $x < t$.} 
Here, $p=\dfrac{x}{t}$. We take
\be
K_i=\sqrt{p q_i}\dfrac{\beta_i}{\alpha}|f_1\lr f_1|+\sqrt{\dfrac{pq_iy_i}{x}} \vert v_i\lr u|,
\ee
which reveals		
\beax
&&K_i |\psi\ra\\
&&=\left(\sqrt{pq_i}\frac{\beta_i}{\alpha}|f_1\lr f_1| + \sqrt{\frac{pq_iy_i}{x}} |v_i\lr u|\right) \left( \alpha |f_1\ra+\sqrt{x}|u\ra\right) \\	
&&=\sqrt{pq_i}\beta_i|f_1\ra+\sqrt{pq_iy_i}|u\ra\\
&&=\sqrt{pq_i}|\phi_i\ra.
\eeax
That is,
\beax
\sum_iK_i\vert\psi\lr\psi|K_i^\dagger
=p\sum_iq_i \vert\phi_i\lr\phi_i|
=p\rho.
\eeax
Thus, the success probability of this trace-nonincreasing operation is
\beax
\tr\left( \sum_iK_i\vert\psi\lr\psi|K_i^\dagger\right)=p,
\eeax
where
\beax
\sum_iK_i^\dagger K_i&=&\frac{p(1-t)}{1-x}|f_1\lr f_1|+\frac{p t}{x}|u\lr u|\\
&=&\frac{x(1-t)}{t(1-x)} |f_1\lr f_1| + |u\lr u| \leqslant I
\eeax
since $x<t$ and $p=x/t$.
Thus, when $x<t$, the optimal probability is
\beax
P(|\psi\ra\to\rho)\geqslant\frac{\mT_g(|\psi\ra)}{\mT_g(\rho)}.
\eeax

\textit{Case 3: $t\leqslant x<1$.}
Let
\be
K_i=a_i|f_1\lr f_1|+\frac{\sqrt{q_i}|\phi_i\ra-\alpha a_i|f_1\ra}{\sqrt{x}} \la u|
\ee
with
$a_i=\frac{\alpha^*}{1-t}\sqrt{q_i}\beta_i$.
Then
\beax
K_i|\psi\ra=\sqrt{q_i}|\phi_i\ra
\eeax
which reveals
\beax
\sum_iK_i|\psi\lr\psi|K_i^\dagger
=\sum_iq_i |\phi_i\lr\phi_i|
=\rho.
\eeax
Let
\be
K_0=a_0|f_1\lr f_1|-\frac{\alpha a_0}{\sqrt{x}}\ket{f_1}\bra{u}+\sum_{j=3}^{d}|e_j\lr e_j|
\ee
with
\beax
a_0=\sqrt{1-\sum_{i}|a_i|^2}=\sqrt{\frac{x-t}{1-t}}.
\eeax
We have
$K_0^\dagger K_0+\sum_iK_i^\dagger K_i=I$,
$\sum_iK_i|\psi\lr\psi|K_i^\dagger=\rho$,
and $K_i|f_1\ra\propto|f_1\ra$ for each $i$. In particular, for the case of $d = 2$, $K_0 = a_0|f_1\lr f_1|-\dfrac{\alpha a_0}{\sqrt{x}} \ket{f_1}\bra{u}$.

Through the above discussion, we conclude that
\beax
P(|\psi\ra\to\rho)\geqslant\min\left\{\dfrac{\mT_g(|\psi\ra)}{\mT_g(\rho)}, \, 1\right\}.
\eeax 
Together with Eq.~\eqref{tra1}, the proof is completed.
\end{proof}

From the arguments above, the following are straightforward.	

\begin{cor}\label{cor2}
The maximum probability for a pure state transformation $|\psi\ra\to|\phi\ra$ via texture-free operations is
\be
P(|\psi\ra\to|\phi\ra)=\min\left\{\dfrac{\mT_g(|\psi\ra)}{\mT_g(|\phi\ra)},\,1\right\}.
\ee	
\end{cor}

\begin{cor}\label{cor3}
For any ensemble $\{p_j,|\phi_j\ra\}_{j=1}^m$ of $\rho$, there exists a pure state
\bea\label{cor3-}
|\psi\ra=\sqrt{1-t}\,|f_1\ra+\sqrt{t}\,|u\ra,
\quad 
\la f_1|u\ra=0
\eea
with 
\be \label{cor3--}
t=1-\sum_{j=1}^m p_j |\la f_1|\phi_j\ra|^2
\ee 
such that $|\la f_1|\psi\ra|^2=\la f_1|\rho|f_1\ra$
and a texture-free operation $\mE$ with Kraus operators
$\{K_j\}_{j=1}^m$ satisfying
\beax
K_j|\psi\ra=\sqrt{p_j}\,|\phi_j\ra,
\qquad j=1,\dots,m .
\eeax
\end{cor}

\subsection{Deterministic transformations for qubit states}\label{deter}

In this section, we consider the conversion of qubit states under texture-free operations.  Recall that every qubit state $\rho$ admits the Bloch representation, i.e.,
\bea\label{bloch}
\rho=\frac{1}{2}(I+\bf{r}\cdot{\bm{\sigma}}),
\eea
where $\bm{\sigma}=(\sigma_{x}, \sigma_{y}, \sigma_{z})$ is the vector of Pauli operators, ${\bf{r}}=(r_x, r_y, r_z)$ is the Bloch vector of $\rho$, $\|\bf{r}\|\leqslant 1$. Analogous to the imaginarity transformation studied in Ref.~\cite{Wu2021pra}, we derive necessary and sufficient conditions for deterministic transformations-not only from pure state to any other state, but also from mixed state to arbitrary target state.

\begin{table*}[ht]
\caption{Comparison of conditions for deterministic transformations $\rho \rightarrow \sigma$ in the resource theory of imaginarity, coherence and texture, where $\mathbf{r}=(r_x,r_y,r_z)$ and $\mathbf{s}=(s_x,s_y,s_z)$ are the Bloch vectors of the $\rho$ and $\sigma$, respectively.}
\label{tab:2}
\begin{ruledtabular}
\begin{tabular}{lll}
	Resource
	& Conditions for deterministic transformation $\rho \rightarrow \sigma$&Reference\\
	\midrule
	Texture &  $s_x\geqslant r_x$ 
	~\text{and}~$(1 - r_x)(1 - |\mathbf{s}|^2) \geqslant (1 - s_x)(1 - |\mathbf{r}|^2)$ & Theorem~\ref{th4}
	\\
	Coherence &  $s_x^2 + s_y^2 \leqslant r_x^2 + r_y^2$ and 
	$s_z^2\leqslant1-\frac{1-r_z^2}{r_x^2+r_y^2}\bigl(s_x^2 + s_y^2\bigr). $ &\cite{Shi2017sr} \\
	Imaginarity &  $s_y^2 \leqslant r_y^2$ and
	$\frac{1-s_y^2-s_x^2}{s_y^2}
	\geqslant
	\frac{1-r_y^2-r_x^2}{r_y^2}$ &\cite{Wu2021pra}
\end{tabular}%
\end{ruledtabular}
\end{table*}

\begin{theorem}\label{th4}
For any qubit states $\rho$ and $\sigma$, the transition $\rho\to\sigma$ is possible via texture-free operations if and only if
\be\label{qu1}
s_x\geqslant r_x 
\ee
and
\be\label{qu2}
(1 - r_x)(1 - |\mathbf{s}|^2) \geqslant (1 - s_x)(1 - |\mathbf{r}|^2),
\ee
where $\mathbf{r}=(r_x,r_y,r_z)$ and $\mathbf{s}=(s_x,s_y,s_z)$ are the Bloch vectors of the initial and the target state, respectively, $|\mathbf{r}|^2 = r_x^2+r_y^2+r_z^2$, $|\mathbf{s}|^2=s_x^2+s_y^2+s_z^2$.	
\end{theorem}

\begin{proof}
``$\Rightarrow$''
Let	$\mE(\rho)=\sigma$ for some texture-free operation $\mE$ and $|\pm\ra=\dfrac{1}{\sqrt{2}}(|0\ra \pm|1\ra)$. Then the Kraus operators of $\mE$ are upper triangular matrices in the basis \(\{|+\ra,|-\ra\}\), i.e.,
\beax
K_j=\begin{pmatrix}
\alpha_j&\beta_j\\
0&\delta_j
\end{pmatrix}.
\eeax	
In additon, in this basis,
\beax
\rho&=&\dfrac{1}{2}
\begin{pmatrix}
1+r_x&r_z+\rmi r_y\\
r_z-\rmi r_y&1 - r_x
\end{pmatrix}
=
\begin{pmatrix}
p&x\\
x^*&1 - p
\end{pmatrix},\\
\sigma   &=&
\begin{pmatrix}
p'&y\\
y^*&1 - p'
\end{pmatrix},
\eeax
where
\beax
p=\dfrac{1+r_x}{2},~ x = \dfrac{r_z+\rmi  r_y}{2},~
p'=\dfrac{1+s_x}{2},~ y = \dfrac{s_z+\rmi s_y}{2}.
\eeax

Note that
$\sum_jK_j^\dagger K_j=I$
is equivalent to
\bea\label{conditions}
\left\lbrace \begin{aligned}
&\sum_j|\alpha_j|^2=1,~
\sum_j\alpha_j^*\beta_j = 0,\\
&\sum_j\left(|\beta_j|^2+|\delta_j|^2\right)=1.
\end{aligned}\right. 
\eea
We define
$C= \sum_j\alpha_j\delta_j^*$,
$D= \sum_j|\delta_j|^2$,
$E= \sum_j\beta_j\delta_j^*$
for convenience.
Then $\sum_jK_j \rho K_j^\dagger=\sigma$ reveals
$1-p'=D(1-p)$ and
$y=Cx+E(1-p)$.
Switching back to the Bloch vectors, we get
\be\label{Bloch1}
{1-s_x=D(1-r_x)}
\ee
and
\be\label{Bloch2}
{s_z+\rmi s_y=C(r_z+\rmi r_y)+E(1-r_x)}.
\ee
Since $0\le D\le1$, it immediately follows that
\beax
1-s_x\le1-r_x\Rightarrow
{s_x\geqslant r_x}.
\eeax		
If $r_x=1$, then $s_x=1$. In this case, both $\rho=\sigma=\tau_f$. Taking the texture-free operation as the identity map, the two states can be transformed and they satisfy Eqs.~\eqref{qu1} and~\eqref{qu2}.

We denote by $\boldsymbol{\alpha}=(\alpha_j)_j$, $\boldsymbol{\beta}=(\beta_j)_j$, and $\boldsymbol{\delta}=(\delta_j)_j$ as complex vectors.
Then the conditions in Eq.~\eqref{conditions}		
is equivalent to
\beax
\|\boldsymbol{\alpha}\|^2=1,~\boldsymbol{\alpha}\perp\boldsymbol{\beta},~ \|\boldsymbol{\beta}\|^2=1-D,~\|\boldsymbol{\delta}\|^2=D.
\eeax

\textit{Case 1: $s_x = r_x \neq 1$.}	
By Eq.~\eqref{Bloch1}, $\sum_j|\beta_j|^2=1-D=0$, which reveals $\beta_j=0$ for each $j$ and $E=0$. This shows that $y=Cx$. 
Since the length of the projection of $\boldsymbol{\delta}$ onto the subspace spanned by $\boldsymbol{\alpha}$  cannot exceed $\|\delta\|$, we have
\beax
|C|^2\leqslant1,
\eeax
namely, 
$(s_y^2+s_z^2)\le(r_y^2+r_z^2)$ as desired.

\textit{Case 2: $s_x>r_x$.}				
Since the length of the projection of $\boldsymbol{\delta}$ onto the subspace spanned by $\boldsymbol{\alpha}$ and $\boldsymbol{\beta}$ cannot exceed $\|\delta\|$, we have
\be\label{Bloch3}
|C|^2+\dfrac{|E|^2}{1-D}\leqslant D.
\ee
Together with
\beax
E=\dfrac{(s_z+\rmi s_y)-C(r_z+\rmi r_y)}{1-r_x},
\eeax 
we obtain
\be\label{CD}
|C|^2+\dfrac{|(s_z+\rmi s_y)-C(r_z+\rmi r_y)|^2}{(1-r_x)^2(1-D)}\leqslant D.
\ee
Now we let
\be\label{notation1}
u := r_z+\rmi r_y,~ v:=s_z+\rmi s_y,~q:=1-r_x,~ q':=1-s_x.
\ee
Then, it follows from Eq.~\eqref{Bloch1} that
\beax
(1-r_x)^2(1-D)=q(q-q').
\eeax
Thus, the above inequality~\eqref{CD} becomes
\beax
|C|^2+\dfrac{|v-Cu|^2}{q(q-q')} \leqslant\dfrac{q'}{q}.
\eeax
Therefore, there exists some $C$ such that the above inequality holds if and only if
\beax
\dfrac{|v|^2}{|u|^2+q(q-q')} \leqslant\dfrac{q'}{q}.
\eeax
which is equivalent to
\beax
(1 - r_x)(1 - |\mathbf{s}|^2) \geqslant (1 - s_x)(1 - |\mathbf{r}|^2).
\eeax

``$\Leftarrow$''
\textit{Case 1: $s_x = r_x$.}
In this case, the inequality~\eqref{qu2} reduces to
\beax
s_y^2+s_z^2\leqslant r_y^2+r_z^2,
\eeax
which is equivalent to
\beax
|v|\leqslant|u|.
\eeax		
If $u=0$, then $v=0$ and the identity channel works directly. If $u\ne0$, we let
$\lambda:=\dfrac{v}{u}$.
In the basis $\{|+\ra, |-\ra\}$, we take
\beax
K_1=\begin{pmatrix}1&0\\0&\lambda^*\end{pmatrix},
~
K_2=\begin{pmatrix}0&0\\0&\sqrt{1-|\lambda|^2}\end{pmatrix}.
\eeax
Then
$K_1^\dagger K_1+K_2^\dagger K_2=I$,
and under this channel, denoted by $\mE$, $x$ is transformed to $\lambda x=y$ while the diagonal elements remains unchanged, that is, $\mE(\rho)=\sigma$.

\textit{Case 2: $s_x>r_x$.}		
Let
\beax
C_0:= \dfrac{u^*v}{|u|^2+q(s_x-r_x)},~ E_0:=\dfrac{v-C_0u}{q},
\eeax
where $u$, $v$, $q$ are defined as in Eq.~\eqref{notation1}. The same as Eq.~\eqref{Bloch3}, we have
\beax
|C_0|^2+\dfrac{|E_0|^2}{1-D}\leqslant D.
\eeax
We take three Kraus operators (still in the basis $\{|+\ra,|-\ra\}$):
\begin{align*}
K_1&=\begin{pmatrix}1&0\\ 0 & C_0^*\end{pmatrix},\\
K_2&=\begin{pmatrix}0&\sqrt{1 - D}\\0&E_0^*/\sqrt{1-D}\end{pmatrix},\\
K_3&=\begin{pmatrix}0&0\\0&\sqrt{D-|C_0|^2-\dfrac{|E_0|^2}{1-D}}\end{pmatrix}.
\end{align*}		
It turns out that
\beax
\sum_{j=1}^3K_j^\dagger K_j  
&=&
\begin{pmatrix}
1&0\\
0&|C_0|^2
\end{pmatrix}
+
\begin{pmatrix}
0&0\\
0&1-D+\dfrac{|E_0|^2}{1-D}
\end{pmatrix} \\
&&+
\begin{pmatrix}
0&0\\
0&D-|C_0|^2- \dfrac{|E_0|^2}{1-D}
\end{pmatrix}
=I 
\eeax
and 
\beax
\sum_{j=1}^3K_j\rho K_j^\dagger
&=&
\begin{pmatrix}
p&xC_0 \\
(xC_0)^*&|C_0|^2(1-p)
\end{pmatrix}\\
&&+
\begin{pmatrix}
(1-D)(1-p)&(1-p)E_0\\
(1-p)E_0^*&\dfrac{|E_0|^2(1-p)}{1-D}
\end{pmatrix} \\
&&+
\begin{pmatrix}
0&0\\
0&\left(D-|C_0|^2-\dfrac{|E_0|^2}{1-D}\right)(1-p)
\end{pmatrix}\\
&=&
\begin{pmatrix}
p+(1-D)(1-p)&(1-p)E_0+xC_0\\
(1-p)E_0^*+(xC_0)^*&D(1-p)
\end{pmatrix} \\
&=&
\begin{pmatrix}
p'&y\\
y^*&1 - p'
\end{pmatrix}
\eeax
as desired.
\end{proof}

At the end of this section, we list in Table~\ref{tab:2} the conditions of deterministic transformation between qubit states in the resource theory of imaginarity, coherence and texture for readers' convenience.


\section{Building texture measures}\label{build}


\subsection{Convex function based measures}

Since the free-state set contains only one element, it is natural to construct a texture measure from the overlap $\la f_1|\rho|f_1\ra$, which has both geometric interpretation and simple algebraic representation~\cite{Parisio2024prl}. The rugosity measure $\mT_R(\rho)=-\ln\la f_1 |\rho|f_1\ra$ is such a case.

Note that $-\ln x$ is a convex function and satisfies items (i)-(ii). This motivates us to generalize $\mT_R$ into a wider class.
Let $\check{h}:[0,1]\rightarrow[0, +\infty)$ be a convex function that also admits items (i)-(ii). We define
\be\label{fun3}
{\mT}_{\check{h}}(\rho)=\check{h}(\la f_1|\rho|f_1\ra).
\ee
The following is the main result of this subsection.

\begin{theorem}\label{th6}
For any convex function $\check{h}:[0,1]\rightarrow[0, +\infty)$ satisfying items (i)-(ii), ${\mT}_{\check{h}}$ defined in Eq.~\eqref{fun3} constitutes a valid texture measure.
\end{theorem}

\begin{proof}
(T1) is clear by definition.

For state $\rho$, we have~\cite{Parisio2024prl}
\bex\label{Sigma}
d\cdot\la f_1|\rho|f_1\ra=\sum_{i,j=0}^{d-1}\la i|\rho|j\ra=\sum_{i,j=0}^{d-1} \rho_{ij}.
\eex
Hereafter, we let $\Sigma(\rho):=\sum_{i,j=0}^{d-1} \rho_{ij}$. Using Eq.~(11) from the Supplemental Material of Ref.~\cite{Parisio2024prl}, we know $\Sigma(\mE(\rho)) \geqslant \Sigma(\rho)$, which implies
$\la f_1|\mE(\rho)|f_1\ra\geqslant\la f_1|\rho|f_1\ra$.
Since $h$ is decreasing, we obtain
$\mT_{\check{h}}(\rho)\geqslant\mT_{\check{h}}(\mE(\rho))$.
That is, (T2) holds.

In addition, for any ensemble $\{p_j, \rho_j\}$, since $\check{h}$ is convex, 
$\mT_{\check{h}}\Big(\sum_jp_j\rho_j\Big)
=\check{h}\Big(\sum_jp_j\la f_1|\rho_j|f_1\ra\Big)
\leqslant\sum_jp_j \check{h}\big(\la f_1|\rho_j|f_1\ra\big)
=\sum_jp_j\mT_{\check{h}}(\rho_j)$,
namely, $\mT_{\check{h}}$ satisfies item (T3).		
\end{proof}

$\mT_{\check{h}}$ recovers $\mT_{\Fi}$, $\mT_{B}$, $\mT_R$, and $\mT_g$ as special cases. We list the corresponding convex functions $\check{h}$ of theses texture measures in the second row in
Table~\ref{tab:comparison}. Theorem~\ref{th6} guarantees that $\mT_{\check{h}}$ defines a valid texture measure in the sense of (T1)-(T3). It is then natural to investigate whether $\mT_{\check{h}}$ also satisfies (T4). As one may expect, it is not a texture monotone in general.

It should be emphasized that the present method is not directly applicable to other quantum resources---e.g., entanglement, coherence, and imaginarity---since none of these resources admits a unique free-state characterization. Indeed, in the case of entanglement, free states form a subset of $\mS$ with nonzero measure~\cite{Zyczkowski1998pra}, whereas for coherence and imaginarity, free states form measure-zero subsets of $\mS$. Hence, most states are resourceful.

\begin{table*}[t]
	\caption{The associated function $h$ or $\check{h}$ for the texture measures.}
	\label{tab:comparison}
	\begin{ruledtabular}
		\begin{tabular}{lcccccccccc}
			$h/\check{h}$  
			& $\mT_{tr}(|\psi\ra)$
			& $\mT_{\Fi}(\rho)$
			& $\mT_B(\rho)$
			& $\mT_{\mu}^{TS}(|\psi\ra)$
			& $\mT_J(\psi\ra)$
			& $\mT_{he}(|\psi\ra)$
			& $\mT_R(\rho)$
			& $\mT_g(\rho)$
			& $\mT_w(|\psi\ra)$
			& $\mT_{\alpha}^{\mathrm{skew}}(|\psi\ra)$\\
			\midrule
			$h$
			& $\sqrt{1-x}$
			& $1-x$
			&
			& $\dfrac{1-x}{1-\mu}$
			& $H_2\left(\dfrac{1+\sqrt{x}}{2}\right)$
			& $2(1-x)$
			&
			& $1-x$
			& {?}
			& {$x(1-x)$}\\
			$\check{h}$
			&
			& $1-x$
			& $2(1-\sqrt{x})$
			&
			&
			&
			& $-\ln x$
			& $1-x$
			&
			&
		\end{tabular}%
		\begin{tablenotes}[flushleft]
			\footnotesize
			\item[]\raggedright
			\setlength{\parfillskip}{0pt}
			{\textit{Note}}: (i) Here, $\mT(\lvert\psi\rangle)$ refers to the corresponding function $h$ is defined on pure states whereas $\mT(\rho)$ refers to the the corresponding function $h$ is defined on both pure states and mixed states uniformly. (ii) ``?'' means it is unknown.
		\end{tablenotes}
	\end{ruledtabular}
\end{table*}

\subsection{Texture cost}\label{measu}

Note that, any state can be generated from some input pure states via texture-free operations~\cite{Parisio2024prl}. For any given state $\rho$, we let $\mP_{\rho}$ be the set of all pure states that can be transformed into $\rho$ by means of texture-free operations. Let $h: [0,1]\rightarrow[0, +\infty)$ be a function satisfying conditions (i)-(iii) and let $\mT$ be defined as in Eq.~\eqref{fun1}. We define the texture cost associated with $\mT$ by
\be
\check{\mT}(\rho)=\min_{|\phi\ra\in\mP_{\rho}}\mT(|\phi\ra). 
\ee

\begin{theorem}
$\check{\mT}(\rho)=h(\la f_1|\rho|f_1\ra)$ and $\check{\mT}(\rho)$ satisfies (T1)-(T2) and (T4).
\end{theorem}

\begin{proof}	
Assume that $\rho=\sum_ip_i|\phi_i\lr\phi_i|$. By Corollary~\ref{cor3}, there exists $|\psi\ra$ such that $h(|\la f_1|\psi\ra|^{2})=h(\la f_1|\rho|f_1 \ra)$ and
\beax
\check{\mT}(\rho)\leqslant h(|\la f_1|\psi\ra|^{2})=h(\la f_1|\rho|f_1 \ra).
\eeax
However, if $|\Psi\ra$ be an optimal state such that $\rho=\mE(|\Psi\lr\Psi|)=\sum_jK_j|\Psi\lr\Psi|K_j^\dagger$ and $\check{\mT}(\rho)=h(|\la f_1|\Psi\ra|^2)$, by Eq.~(25) in Ref.~\cite{Cui2025arvix}, we have $ \sum_j|\la f_1|K_j|\Psi\ra|^2\geqslant |\la f_1|\Psi\ra|^2$, and thus
\beax
h\left( \sum_j|\la f_1|K_j|\Psi\ra|^2\right) \leqslant h\left( |\la f_1|\Psi\ra|^2\right).
\eeax
It follows that
\beax
\check{\mT}(\rho)
&=&h(|\la f_1|\Psi\ra|^2)
\geqslant h\left( \sum_j|\la f_1|K_j|\Psi\ra|^2\right)\\
&=&h(\la f_1|\rho|f_1 \ra).
\eeax
Therefore, $\check{\mT}(\rho)=h(\la f_1|\rho|f_1\ra)$.

Since 
\beax
\check{\mT}(\tau_f)=\min_{|\phi\ra\in\mP_{\tau_f}}\mT(|\phi\ra),
\eeax
we take $|\phi\ra=|f_1\ra$. Then $\check{\mT}(\tau_f)=0$. So $\check{\mT}(\rho)$ satisfies (T1). 
For any state $\rho$, we let $|\psi\ra\in\mP_{\rho}$ is the optimal pure state such that $\check{\mT}(\rho)=\mE_0(|\psi\ra\la\psi|)$ for some texture-free operation $\mE_0$. That is
\beax
\check{\mT}(\rho)=\mT(|\psi\ra),~
\mE_0(|\psi\lr\psi|)=\rho.
\eeax
Applying any texture-free operation $\mE$ on $\rho$, we have $\mE\bigl(\mE_0(|\psi\lr\psi|)\bigr)=\mE(\rho)$. Note that $\mE\circ\mE_0$ is also a texture-free operation, by the definition of $\check{\mT}$, we immediately obtain
\beax
\check{\mT}(\rho)=\mT(|\psi\ra)\geqslant \check{\mT}(\mE(\rho)),
\eeax
which reveals the monotonicity under texture-free operations.

We now check item (T4). Since $\check{\mT}(|\psi\ra)=\mT(|\psi\ra)$ for all pure states $|\psi\ra$, strong monotonicity holds for pure states. Let $\rho$ be a mixed state and $\mE(\cdot)=\sum_nK_n(\cdot)K_n^\dagger$ be a texture-free operation. We write
\beax
p_n=\tr (K_n\rho K_n^\dagger),\quad\rho_n=\frac{K_n \rho K_n^\dagger}{p_n},
\eeax
and let $|\psi\ra$ be the optimal state in $\mP_{\rho}$ such that $\check{\mT}(\rho)=\mT(|\psi\ra)=h(|\la f_1|\psi\ra|^2)$. We assume that $\rho = \sum_jA_j |\psi\lr\psi|A_j^\dagger$ and let
\begin{align*}
c_{nj}&=\la\psi|K_n^\dagger A_j^\dagger A_jK_n|\psi\ra,~
|\psi_{nj}\ra=\dfrac{K_nA_j|\psi\ra}{\sqrt{c_{nj}}},~
q_{nj}=\dfrac{c_{nj}}{p_n}.
\end{align*}
Then, $\rho_n=\sum_{j}q_{nj}|\psi_{nj}\lr\psi_{nj}|$.
For any $\rho_n$, there exists a pure state $|\eta_n\ra\in\mP_{\rho_n}$ for each $n$. Clearly, 
$\check{\mT}(\rho_n)\leqslant\mT(|\eta_n\ra)$.
Together with $\mT(|\eta_n\ra)=h(|\la f_1|\eta_n\ra|^2)$ and Eq.~\eqref{cor3--}, we obtain
\be\label{ctrn}
\check{\mT}(\rho_n)\leqslant h\left( \sum_jq_{nj}|\la f_1|\psi_{nj}\ra |^2\right).
\ee
Therefore,
\beax
&&\sum_np_n \check{\mT}\left(\rho_n\right)\\
\leqslant&& \sum_np_nh\left( \sum_jq_{nj}|\la f_1|\psi_{nj}\ra |^2\right)\\
=&& \sum_np_nh\left(\frac{\sum_j\la f_1|K_n A_j|\psi\ra\la\psi|A_j^\dagger K_n^\dagger|f_1\ra}{p_n}\right)\\
\leqslant&& h\left( \sum_{n,j}|\la f_1|K_nA_j|\psi\ra|^2 \right),
\eeax
where the first inequality follows directly from Eq.~\eqref{ctrn} and the second inequality follows from the concavity of $h$. 
On the other hand, by Eq.~(25) in Ref.~\cite{Cui2025arvix}, we have $ \sum_{n,j}|\la f_1|K_nA_j|\psi\ra|^2\geqslant |\la f_1|\psi\ra|^2$, and thus 
\be\label{hls}
h\left( \sum_{n,j}|\la f_1|K_nA_j|\psi\ra|^2 \right)
\leqslant h(|\la f_1|\psi\ra|^2)=\check{\mT}(\rho)
\ee
since $h$ is monotonic decreasing. 
This complete the proof of the strong monotonicity of $\check{\mT}$. 
\end{proof}

We need note here that, if $\check{\mT}$ is not convex, it is not a texture measure.
A similar situation also appears in the resource theories of entanglement~\cite{Shi2021anpb}, coherence~\cite{Yu2020pra}, and imaginarity~\cite{Du2025pra}, where the measures constructed through pure-state transformations satisfy non-negativity, monotonicity and strong monotonicity. For these three types of quantum resources, measures of this type require additional conditions to satisfy convexity. We now turn to the question: under what conditions does $\check{\mT}$ satisfy (T3)? The following theorem provides an equivalent description of this issue.

\begin{theorem}\label{convex1}
The following statements are equivalent:
\begin{enumerate}
\item[(i)] $\check{\mT}$ is a texture monotone;
\item[(ii)] $\check{\mT}(\rho)=\mT_F(\rho)$;
\item[(iii)] $h$ is an affine function.
\end{enumerate}		
\end{theorem}

\begin{proof}
$(i) \Rightarrow (ii)$. Assume that $\rho=\sum_ip_i|\phi_i\lr\phi_i|$. Since $\check{\mT}$ is convex,  $\check{\mT}(\rho)=\check{\mT}(\sum_ip_i|\phi_i\lr\phi_i|)\leqslant\sum_ip_i\check{\mT}(|\phi_i\ra)$. The above inequality holds for any ensemble decomposition of $\rho$, we have $\check{\mT}(\rho)\leqslant\min_{p_i,|\phi_i\ra}\sum_i{p_i}h(|\la f_1|\phi_i\ra|^2)=\mT_F(\rho)$. On the other hand, 
\be\label{jen}
h(\la f_1|\rho|f_1\ra)\geqslant\min_{p_i,|\phi_i\ra}\sum_i{p_i}h(|\la f_1|\phi_i\ra|^2)=\mT_F(\rho) 
\ee
since $h$ is concave. That is $\check{\mT}(\rho)=\mT_F(\rho)$.

$(ii) \Rightarrow (iii)$. If $\check{\mT}(\rho)=\mT_F(\rho)$, then $\check{\mT}$ is a texture measure, so it is convex, which implies $h$ is convex. But $h$ is concave by definition. So it is affine.

$(iii) \Rightarrow (i)$. If $h$ is affine, then it is convex, and thus $\check{\mT}$ is convex. This implies $\check{\mT}$ is a texture monotone.
\end{proof}

Theorem~\ref{convex1} provides a simple way to construct texture monotones. However, one can easily check that such a affine function $h$ is uniquely determined by $h(x) = k(x-1)$ with $k< 0$.
Choosing $h(x)=1-x$, it recovers the texture monotone $\mT_{\Fi}$, and all other cases are only multiples of $\mT_{\Fi}$. Obviously, Theorem~\ref{convex1} enables us to verify whether $\mT_B$ and $\mT_R$ obey strong monotonicity. As their associated functions $\check{h}$ are all non-affine, none of them satisfy the strong monotonicity.

\begin{table*}
	\caption{Property of texture measures.}
	\label{tab:property}
	\begin{ruledtabular}
		\begin{tabular}{lcccccccccccc}
			Texture measure
			& $\mT_{tr}$
			& $\mT_{\Fi}$
			& $\mT_B$
			& $\mT_r$
			& $\mT_{\mu}^{TS}$
			& $\mT_{he}$
			& $\mT_J$
			& $\mT_R$
			& $\mT_g$
			& $\mT_w$
			& $\mT_{rob}$
			& $\mT_{\alpha}^{\mathrm{skew}}$\\
			\midrule
			(T1)
			& $\checkmark$\cite{Wang2025pra}
			& $\checkmark$\cite{Wang2025pra}
			& $\checkmark$\cite{Wang2025pra}
			& $\checkmark$\cite{Wang2025pra}
			& $\checkmark$\cite{Zhang2025pla}
			& $\checkmark$\cite{Muthuganesan2026pla}
			& $\checkmark$\cite{Muthuganesan2026pla}
			& $\checkmark$\cite{Parisio2024prl}
			& $\checkmark$\cite{Wang2025pra}
			& $\checkmark$\cite{Zhang2025pla}
			& $\checkmark$\cite{Wang2025pra}
			& $\checkmark$\cite{Muthuganesan2026pla}\\
			(T2)
			& $\checkmark$\cite{Wang2025pra}
			& $\checkmark$\cite{Wang2025pra}
			& $\checkmark$\cite{Wang2025pra}
			& $\checkmark$\cite{Wang2025pra}
			& $\checkmark$\cite{Zhang2025pla}
			& $\checkmark$\cite{Muthuganesan2026pla}
			& $\checkmark$\cite{Muthuganesan2026pla}
			& $\checkmark$\cite{Parisio2024prl}
			& $\checkmark$\cite{Wang2025pra}
			& $\checkmark$\cite{Zhang2025pla}
			& $\checkmark$\cite{Wang2025pra}
			& $\times$\\
			(T3)
			& $\checkmark$\cite{Wang2025pra}
			& $\checkmark$\cite{Wang2025pra}
			& $\checkmark$\cite{Wang2025pra}
			& $\checkmark$\cite{Wang2025pra}
			& $\checkmark$\cite{Zhang2025pla}
			& $\checkmark$\cite{Muthuganesan2026pla}
			& $\checkmark$\cite{Muthuganesan2026pla}
			& $\checkmark$\cite{Parisio2024prl}
			& $\checkmark$\cite{Wang2025pra}
			& $\checkmark$\cite{Zhang2025pla}
			& $\checkmark$\cite{Wang2025pra}
			& $\checkmark$\cite{Muthuganesan2026pla}\\
			(T4)
			& $\times$
			& $\checkmark$\cite{Cui2025arvix}
			& $\times$
			& $\checkmark$\cite{Wang2025pra}
			& $\times$
			& $\times$
			& $\times$
			& $\times$
			& $\checkmark$\cite{Cui2025arvix}
			& $\checkmark$\cite{Zhang2025pla}
			& ?
			& $\times$\\
			(T5)
			& $\times$
			& $\times$
			& $\times$
			& $\times$
			& $\times$
			& $\times$
			& $\times$
			& $\times$
			& $\times$
			& $\times$\cite{Zhang2025pla}
			& $\times$
			& $\times$
		\end{tabular}%
	\begin{tablenotes}[flushleft]
		\footnotesize
		\item[]
		\parbox{\textwidth}{
		\raggedright
		\setlength{\parfillskip}{0pt}
	\textit{Note:}{  (i) Some of the corresponding properties of ``$\times$'' that are not accompanied by references are demonstrated in Appendix, while the others can be directly derived from the logical relations among conditions (T1)-(T5). (ii) It remains unclear whether $\mT_{rob}$ fulfills the strong monotonicity condition.
	 }}
	\end{tablenotes}
	\end{ruledtabular}
\end{table*}

\subsection{Convex-roof extended texture measure}

The convex-roof extended measure is the most widely used one in the context of the convex quantum resource theories~\cite{Vidal2000jmo,Du2015qic,Du2025pra,Zhu2017pra,Cui2025arvix}, such as entanglement, coherence, imaginarity, and texture. In such a scenario, for any given resource measure defined on pure states, it induces an associated function that admits some special properties and can be extended into mixed states. In convex quantum resource theories-including entanglement~\cite{Vidal2000jmo}, coherence~\cite{Yu2020pra,Du2015qic}, imaginarity~\cite{Du2025pra}, and texture~\cite{Cui2025arvix}---the convex-roof extended measures are strong resource monotones. Conversely, restricting any strong resource monotone to pure states yields a function that always shares the same distinguishing properties. This is exemplified by the resource theories of entanglement~\cite{Vidal2000jmo}, coherence~\cite{Du2015qic}, and imaginarity~\cite{Du2025pra}.
As in the case of other strong resource monotones, we show below that, for the case of texture, any texture monotone can also induce a corresponding function that admits the special properties.

\begin{theorem}\label{th5}
Let $\mT$ be a texture monotone. Then, for pure states, it is identical to Eq.~\eqref{fun1} for some $h$ satisfying items (i)-(iii).
\end{theorem}

\begin{proof}		
Let $\mT$ be a texture monotone and $x =|\la f_1|\psi\ra|^{2}$. We define a function $h:[0,1]\rightarrow[0, +\infty)$ by
\beax h(x) = \mT(|\psi\lr\psi|).\eeax  
It admits item (i) clearly from conditions (T1). 

Suppose that $x,y\in[0,1]$ and $x\geqslant y$.
We can find two pure states $|\varphi_1\ra,|\varphi_2\ra$ such that
$|\la f_1|\varphi_1\ra|^2=x$ and $|\la f_1|\varphi_2\ra|^2=y$.
By Corollary~\ref{cor2}, there exists a texture-free operation $\mE$ such that $\mE(\varphi_2)=\varphi_1$.
By item (T2), we get
$$
\mT\big(|\varphi_1\lr\varphi_1|\big)\leqslant \mT\big(|\varphi_2\lr\varphi_2|\big).
$$
That is, $h(x)\leqslant h(y)$, namely, $h$ is decreasing on $[0,1]$.

To prove item (iii), for any $x, y\in [0,1]$ and $\lambda \in [0,1]$, we consider two cases. 

\textit{Case 1: $\lambda x+(1-\lambda)y=0$ or $\lambda x+(1-\lambda)y=1$.} In this case, it is straightforward to verify that condition (iii) holds.

\textit{Case 2: Otherwise.}	Let
\beax		
|\psi_1\ra&=a_1|1\ra+a_2|2\ra,\\
|\psi_2\ra&=b_1|1\ra+b_2|2\ra,\\
|\phi\ra&=c_1|1\ra+c_2|2\ra,\\
\eeax
with
\begin{align*}
a_1&=\frac{\sqrt{2x}+\sqrt{2-2x}}{2},\\
a_2&=\dfrac{\sqrt{2x}-\sqrt{2-2x}}{2},\\
b_1&=\frac{\sqrt{2y}+\sqrt{2-2y}}{2},\\
b_2&=\dfrac{\sqrt{2y}-\sqrt{2-2y}}{2},\\
c_1&=\frac{\sqrt{2[\lambda {x}+(1-\lambda)
		{y}]}+\sqrt{2-2[\lambda {x}+(1-\lambda)
		{y}]}}{2},\\
c_2&=\frac{\sqrt{2[\lambda {x}+(1-\lambda)
		{y}]}-\sqrt{2-2[\lambda {x}+(1-\lambda)
		{y}]}}{2}.
\end{align*}
Then $c_1^2-c_2^2\neq0$.
By direct calculation, $|\la f_1|\psi_1\ra|^2=x$,
$|\la f_1|\psi_2\ra|^2=y$,
$|\la f_1|\phi\ra|^2=\lambda {x}+(1-\lambda){y}$.
Taking
\beax
K_1=\begin{pmatrix}
x_1&z_1 \\
z_1&x_1
\end{pmatrix},~
K_2=\begin{pmatrix}
x_2&z_2\\
z_2&x_2
\end{pmatrix},
\eeax
with 
\beax
\begin{cases}
x_1=\dfrac{\sqrt{\lambda}a_1c_1-\sqrt{\lambda}a_2c_2}{c_1^2-c_2^2},\\
z_1=\dfrac{\sqrt{\lambda}a_2c_1-\sqrt{\lambda}a_1c_2}{c_1^2-c_2^2},\\
x_2=\dfrac{\sqrt{1-\lambda}b_1c_1-\sqrt{1-\lambda}b_2c_2}{c_1^2-c_2^2},\\
z_2=\dfrac{\sqrt{1-\lambda}b_2c_1-\sqrt{1-\lambda}b_1c_2}{c_1^2-c_2^2}.
\end{cases}
\eeax
One can verify that $\sum_{j=1}^2 K_j^\dagger K_j = I$.
Therefore, $\mE(\cdot)=K_1(\cdot) K_1^\dagger +K_2(\cdot) K_2^\dagger$ is a texture-free operation. 
In addition,
\beax
\begin{pmatrix}
x_1&z_1\\
z_1&x_1
\end{pmatrix}
\begin{pmatrix}
c_1\\
c_2
\end{pmatrix}
=
\sqrt{\lambda}
\begin{pmatrix}
a_1\\
a_2
\end{pmatrix}
\eeax			
and
\beax			
\begin{pmatrix}
x_2&z_2 \\
z_2&x_2
\end{pmatrix}
\begin{pmatrix}
c_1\\
c_2
\end{pmatrix}
=
\sqrt{1-\lambda}
\begin{pmatrix}
b_1\\
b_2
\end{pmatrix}.
\eeax
Namely, $\mE(|\phi\ra\la\phi|)=\lambda|\psi_1\ra\la\psi_1|+(1-\lambda)|\psi_2\ra\la\psi_2|$.
From (T4), we obtain 
\beax
\lambda \mT(|\psi_1\ra)+(1-\lambda)\mT(|\psi_2\ra)\leqslant \mT(|\phi\ra).
\eeax
That is, $h[\lambda x+(1-\lambda)y]\geqslant\lambda h(x)+(1-\lambda)h(y)$, namely, $h$ is concave.				
\end{proof}

Only the geometric texture measure $\mT_g$ is a convex-roof extended texture measure so far. Of course, we can construct other convex-roof extended texture measures. For example, if we take $h(x)=\sqrt[3]{1-x}$, the corresponding $\mT_F$ defined by Eqs.~\eqref{fun1} and~\eqref{fun2} is a well-defined convex-roof extended texture measure. By Theorem~\ref{th5}, all texture monotones correspond to some non-negative function $h$ satisfying items (i)-(iii). For more clarity, we list them in Table~\ref{tab:comparison}, where the associated functions $h$ can be compared explicitly.

We now obtain different classes of texture measures. We find that $\mT_{\Fi}$ (or equivalently $\mT_g$) belongs to different classes simultaneously, or equivalently, it can be constructed via different approaches. For simplicity, we compare these different classes of texture measures in Table~\ref{tab:existing}.

\begin{table*}
	\caption{Classification of texture measures.}
	\label{tab:existing}
	\begin{ruledtabular}
		\begin{tabular}{lcccccccccccc}
			Texture measure
			& $\mT_{tr}$
			& $\mT_{\Fi}$
			& $\mT_B$
			& $\mT_r$
			& $\mT_{\mu}^{TS}$
			& $\mT_{he}$
			& $\mT_J$
			& $\mT_R$
			& $\mT_g$
			& $\mT_w$
			& $\mT_{rob}$
			& $\mT_{\alpha}^{\mathrm{skew}}$\\
			\midrule
			$\mT$ via contractive distance
			& $\checkmark$
			& $\checkmark$
			& $\checkmark$
			& $\checkmark$
			& $\checkmark$
			& $\checkmark$
			& $\checkmark$
			& 
			& $\checkmark$
			& 
			& 
			&\\
			${\mT}_F$
			& 
			& $\checkmark$
			& 
			& 
			& 
			& 
			& 
			& 
			& $\checkmark$
			& 
			&
			& \\
			${\mT}_{\check{h}}$
			& 
			& $\checkmark$
			& $\checkmark$
			& 
			& 
			& 
			& 
			& $\checkmark$
			& $\checkmark$
			& 
			&
			& \\
			$\check{\mT}$
			& 
			& $\checkmark$
			& 
			& 
			& 
			& 
			& 
			& 
			& $\checkmark$
			& 
			& 
			&\\
			Others
			& 
			& 
			& 
			& 
			& 
			& 
			& 
			& 
			& 
			& $\checkmark$
			& $\checkmark$
			& {$\checkmark$}
		\end{tabular}%
\end{ruledtabular}
\end{table*}


\section{conclusion}\label{concl}


We develop a resource-theoretic framework for quantum-state texture and systematically investigate state transformations under texture-free operations in higher dimensions. Drawing on quantification methods for other quantum resources, we introduce the texture cost and discuss the converse theorem for convex-roof extended texture measures. It is worth noting that, for other quantum resources, convex-function-based measures are generally incompatible. We also analyze the basis properties (T1)-(T5) for all the proposed measures.

Overall, together with existing results in the literature, our work establishes a systematic framework for characterizing, quantifying, and manipulating quantum-state texture. An important open problem remains---namely, deriving necessary and sufficient conditions for deterministic transformations between arbitrary mixed states---which calls for future investigation.

\begin{acknowledgements}
Y. G. is supported by the National Natural Science Foundation of China under Grant Nos.~12471434, the Program for Young Talents of Science and Technology in Universities of Inner Mongolia Autonomous Region under Grant No. NJYT25010, the High-Level Talent Research Start-up Fund of Inner Mongolia University under Grant No. 10000-A260015/501, and the Inner Mongolia Autonomous Region
Science and Technology Plan Projects under Grant
No. 2025KYPT0098. F. H is supported by Natural Science Foundation of Inner Mongolia nuder Grant No. 2026MS0238. S. D. is supported by the National Natural Science Foundation of China under Grant No.~12271452.
\end{acknowledgements}

\appendix


\section*{Appendix: Further properties of texture measures}\label{A}


We provide here a number of examples that demonstrate the failure of certain texture measures to satisfy strong monotonicity or direct-sum additivity, thus establishing the properties recorded in Table~\ref{tab:property}.

\begin{example}\label{strong}
Consider the case of dimension $d=3$. We extend $|f_1\ra$ to form an orthonormal basis $\big\{|f_1\ra,|e_2\ra,|e_3\ra\big\}$ and take a free operation $\mE$ with Kraus operators
\beax
K_0 = |f_1\ra\la f_1| + |e_2\ra\la e_2|,~ K_1 = |e_3\ra\la e_3|.
\eeax
Let
\beax
|\psi\ra = \dfrac{|f_1\ra + |e_2\ra}{\sqrt{2}},\quad \rho = \dfrac{1}{2}|\psi\ra\la\psi| + \dfrac{1}{2}|e_3\ra\la e_3|. 
\eeax
Then
\begin{align*}
p_0&=\tr\big(K_0\rho K_0^\dagger\big)=\dfrac{1}{2},
&
p_1&=\tr\big(K_1\rho K_1^\dagger\big)=\dfrac{1}{2},\\
\rho_0&=\dfrac{K_0\rho K_0^\dagger}{p_0}=|\psi\ra\la\psi|,
&
\rho_1&=\dfrac{K_1\rho K_1^\dagger}{p_1}=|e_3\ra\la e_3|.
\end{align*}
The eigenvalues of $\rho-\tau_f$ are $\dfrac{-1+\sqrt{5}}{4}, \dfrac{-1-\sqrt{5}}{4}, \dfrac{1}{2}$.
Hence,
\beax
\mT_{tr}(\rho)=\dfrac{1}{2}\|\rho-\tau_f\|_\tr=\dfrac{1+\sqrt{5}}{4}.
\eeax
On the other hand, direct calculation gives
\beax
\mT_{tr}(\rho_0)=\dfrac{1}{\sqrt2},\quad \mT_{tr}(\rho_1)=1.
\eeax
This leads to $\sum_{j=0}^1 p_j \mT_{tr}(\rho_j)=\dfrac{1}{2}\left(\dfrac{1}{\sqrt{2}} + 1\right)\approx 0.8536>\dfrac{1+\sqrt{5}}{4} \approx 0.8090$, namely, $\mT_{tr}$ violates the strong monotonicity.

For $\mT_J$, we have $\mT_J(\rho) \approx 0.7238<\sum_{j=0}^1 p_j \mT_J(\rho_j)=\frac12 H_2\left(\dfrac{2+\sqrt{2}}{4}\right)+\frac12 H_2\left(\dfrac{1}{2}\right)\approx 0.8004$, i.e., the strong monotonicity fails either.

For the Tsallis relative entropy measure $\mT_{\mu}^{TS}$,
\beax
\mT_{\mu}^{TS}(\rho_0)&=&\dfrac{1}{2(1-\mu)},  \\
\mT_{\mu}^{TS}(\rho_1)&=&\dfrac{1}{1-\mu},  \\
\mT_{\mu}^{TS}(\rho)&=&\dfrac{1-2^{-\mu-1}}{1-\mu}.
\eeax
Note that $2^{-\mu-1}>\dfrac{1}{4}$ for any $\mu\in(0,1)$, this implies $1-2^{-\mu-1}<\dfrac{3}{4}$. Thus ${\mT_{\mu}^{TS}(\rho)< \sum_{j=0}^1 p_j \mT_{\mu}^{TS}(\rho_j)}$. 

For the measure $\mT_{he}$, it is straightforward that
\beax
\mT_{he}(\rho_0)=1, ~
\mT_{he}(\rho_1)=2, ~
\mT_{he}(\rho)&=2-\dfrac{1}{\sqrt{2}}.
\eeax
Clearly, $\mT_{he}(\rho)<\sum_{j=0}^1 p_j\mT_{he}(\rho_j)$.
\end{example}

We now show that the $\mT_{\alpha}^{\mathrm{skew}}$ satisfies neither monotonicity nor strong monotonicity under free operations, namely it is not a well-defined texture measure.

\begin{example}\label{2}
For the qubit system, we let $\{|f_1\ra, |e_2\ra\}$ be the reference basis. Take a free operation $\mE$ with Kraus operators
\beax
K_0 = |f_1\ra\la f_1|,~K_1 = |+\ra\la e_2|,
\eeax
where $|+\ra=\dfrac{|f_1\ra+|e_2\ra}{\sqrt{2}}$.
We choose the input state $\rho=|e_2\ra\la e_2|$. Then
\beax
\mT_{\alpha}^{\rm skew}(\rho)
&=&
\la f_1|\rho|f_1\ra
-
\la f_1|\rho^\alpha|f_1\ra
\la f_1|\rho^{1-\alpha}|f_1\ra \\
&=& 0.
\eeax
But $\mE(\rho)=|+\ra\la +|$. Hence, $\mT_{\alpha}^{\rm skew}\bigl(\mE(\rho)\bigr)=\dfrac{1}{4}$, i.e., $\mT_{\alpha}^{\rm skew}$ does not satisfy monotonicity under free operations. As monotonicity follows from the strong monotonicity and the convexity, the failure of monotonicity implies that $\mT_{\alpha}^{\rm skew}$ necessarily fails to obey the strong monotonicity.
\end{example}

At last, we show below that  $\mT_{\Fi}$, $\mT_r$, $\mT_{rob}$ and $\mT_{\alpha}^{\rm skew}$ do not satisfy the direct-sum additivity condition (T5).

\begin{example}\label{direct}
Let $\rho_1=\rho_2=\tau_f$ be a qubit state, $p=\dfrac{1}{2}$. Let $\rho=\dfrac{1}{2}\rho_1\oplus\dfrac{1}{2}\rho_2$. We have
\begin{align*}
&\mT_{\Fi}\left(\dfrac{1}{2}\rho_1\oplus\dfrac{1}{2}\rho_2\right)=\dfrac{1}{2}, ~
\mT_r\left(\dfrac{1}{2}\rho_1\oplus\dfrac{1}{2}\rho_2\right)=\infty, \\
&\mT_{rob}\left(\dfrac{1}{2}\rho_1\oplus\dfrac{1}{2}\rho_2\right)=\infty.
\end{align*}
This demonstrates that $\mT_{\Fi}$, $\mT_r$, and $\mT_{rob}$ fail to satisfy the direct-sum additivity condition (T5). 

Taking $\sigma_2=|-\ra\la-|$ with $|-\ra=\dfrac{1}{\sqrt{2}}(|0\ra-|1\ra)$, we can obtain
\beax
&\mT_{\alpha}^{\rm skew}\left(\dfrac{1}{2}\rho_1\oplus\dfrac{1}{2}\sigma_2\right)=\dfrac{1}{8},\\
&\dfrac{1}{2}\mT_{\alpha}^{\rm skew}(\rho_1)+\dfrac{1}{2}\mT_{\alpha}^{\rm skew}(\sigma_2)=0.
\eeax
Therefore, $\mT_{\alpha}^{\rm skew}$ does not satisfy condition (T5).
\end{example}

\end{document}